\documentclass[10pt,journal]{IEEEtran}
\usepackage[letterpaper, top=0.5in, bottom=0.5in, left=0.70in, right=0.70in, columnsep= 0.25in]{geometry}
\IEEEoverridecommandlockouts

\usepackage[utf8]{inputenc}
\usepackage[T1]{fontenc}
\usepackage{url}
\usepackage{ifthen}
\usepackage{cite}
\usepackage{amsmath} 
\usepackage{physics}
\usepackage{array}
\usepackage{nicematrix}                          
\usepackage{amsthm}
\usepackage{amsfonts}
\usepackage{mathrsfs}
\usepackage{multirow}
\usepackage{citesort}
\usepackage{tikz}
\usetikzlibrary{arrows.meta}
\usepackage{pgfplots}
\pgfplotsset{compat=1.18}
\usepackage{bm}
\usepackage{algorithm}
\usepackage{algorithmic}      
\usepackage{graphicx}
\usepackage{epstopdf}
\usepackage{amssymb}
\usepackage{enumitem}
 \usepackage{CJKutf8}

\usepackage{diagbox}
\usepackage{makecell}
\usepackage{siunitx}
\usepackage{bigstrut}
\usepackage{tikz}
\usepackage{multirow}
\usepackage{booktabs}
\usepackage{stmaryrd}
\usepackage[caption=false]{subfig}
\usepackage{amsmath,amssymb,mathtools}

\usepackage{xhfill}
\allowdisplaybreaks[4]
\newcommand{\xfilll}[2][1ex]{%
	\dimen0=#2\advance\dimen0 by #1%
	\leaders\hrule height \dimen0 depth -#1\hfill%
}

\newtheoremstyle{thry}
{3pt}
{3pt}
{}
{1em}
{}
{:}
{.5em}
{}
\theoremstyle{thry}
\newtheorem{proposition}{\emph{\textbf{Proposition}}}
\newtheorem{theorem}{\emph{\textbf{Theorem}}}

\newtheorem{example}{\emph{Example}}
\IfFileExists{Figures/concatenate.tex}
{
  
}
{
  
}

\ifCLASSOPTIONonecolumn

\fi

\ifCLASSOPTIONtwocolumn

\fi

\begin{document}
\title{On the Weight Distribution of Concatenated Code Ensemble Based on the Plotkin Construction
} \author{
\IEEEauthorblockN{Xiao Ma,~\IEEEmembership{Member,~IEEE}
\thanks{Xiao Ma is with the School of Computer Science and Engineering, Sun Yat-sen University, Guangzhou 510006, China (e-mail: maxiao@mail.sysu.edu.cn). } 
}}

\maketitle
\pagestyle{empty}
\thispagestyle{empty}

\begin{abstract}
In this note, we reveal a relation between the weight distribution of a concatenated code ensemble
based on the Plotkin construction and those of its component codes.
The relation may find applications in the calculation of the ensemble weight distributions for many codes, including Reed–Muller (RM)-like codes.
\end{abstract}


\section{Introduction}
In 1960, to derive bounds on block codes, Plotkin proposed a method~\cite{Plotkin1960} that combines two codes of length $n$ into a code of length $2n$ as follows.
Let $\mathscr{C}_0$ and $\mathscr{C}_1$ be two binary block codes of length $n$.
Define $\mathscr{C} = \{ \bm{c} = (\bm{u} + \bm{v}, \bm{v}) \mid \bm{u} \in \mathscr{C}_0, \bm{v} \in \mathscr{C}_1 \}$. The following proposition is a well-known fact, see~\cite{LinCostello2004}.

\begin{proposition}
    Let  $\mathscr{C}_0 \in \mathbb{F}_2^n$ and $\mathscr{C}_1 \in \mathbb{F}_2^n$ be two binary block codes with the minimum Hamming distances $d_0$ and $d_1$, respectively.
  Then the concatenated code $\mathscr{C} = \{ \bm{c} = (\bm{u} + \bm{v}, \bm{v}) \mid \bm{u} \in \mathscr{C}_0, \bm{v} \in \mathscr{C}_1 \}$ based on the Plotkin construction has the minimum Hamming distance
  \[
    d=\min\{\,d_0,\; 2d_1\,\}.
  \]
\end{proposition}

\begin{proof}
It is omitted here.
\end{proof}

An immediate question arises: For binary linear block codes, can we calculate the weight distribution of $\mathscr{C}$ from those of $\mathscr{C}_0$ and $\mathscr{C}_1$?
This is generally not solvable since knowing the weight distributions of the component codes alone is not sufficient to determine the concatenated code, and hence its weight distribution. However, it is feasible to calculate the ensemble weight distributions if we properly define a code ensemble from $\mathscr{C}_0$ and $\mathscr{C}_1$.
This note
serves to build a recursion from the weight distributions of $\mathscr{C}_0$ and
$\mathscr{C}_1$ to that of $\mathscr{C} = \{ \bm{c} = (\bm{u} + \bm{v}\Pi, \bm{v}) \mid \bm{u} \in \mathscr{C}_0, \bm{v} \in \mathscr{C}_1 \}$, where $\Pi$ is a uniformly distributed random permutation matrix of order $n$.
This recursion may find applications in the calculation of the weight distributions for some randomly constructed codes and performance evaluations of Reed–Muller (RM) codes~\cite{Muller1954,Reed1954} and polar codes~\cite{Arikan2009}, to some extent.




\section{The Main Result}

\begin{theorem}
  \label{Thm:main}
  Let $\mathscr{C}_0$ and $\mathscr{C}_1$ be two binary linear codes of length $n$
  with weight distributions $A^{(0)}(X)$ and $A^{(1)}(X)$, respectively, where
  \[
    A^{(i)}(X) = \sum_{j=0}^{n} A^{(i)}_j X^j, \quad i=0,1,
  \]
  and $A^{(i)}_j$ is the number of codewords of weight $j$ in $\mathscr{C}_i$.
  Then the weight distribution of
  \[
    \mathscr{C} = \{\,  (\bm{u} + \bm{v}\Pi, \bm{v}) \;\mid\; \bm{u}\in\mathscr{C}_0,\ \bm{v}\in\mathscr{C}_1 \,\}
  \]
  with a uniformly distributed random permutation matrix $\Pi$ is given by
  \[
    A(X) = \sum_{w} A_w X^w,
  \]
where
\begin{align}\label{Eq:theorem}
A_w &= \sum_{w_1=\max\{0,w-n\}}^{\min\{w,n\}} \sum_{i=\max\{0,\,w-n\}}^{\min\{w_1,\,w-w_1\}}
a(w_1,i)  A^{(1)}_{w_1} A^{(0)}_{\,w-2i} ,
\end{align}
and
\begin{equation}
  a(w_1,i) = \frac{\binom{n}{\,w-w_1\,}\binom{\,w-w_1\,}{\,i\,}\binom{\,n-w+w_1\,}{\,w_1-i\,}}{\binom{n}{\,w_1\,}\binom{n}{\,w-2i\,}}.
\end{equation}

Before proving the theorem, we verify that the calculation of ${A}_w$ for $0 \le w < d = \min\{d_0, 2d_1\}$ is consistent with Proposition 1.
Obviously, $A_0 = 1$, since $i = w_1 = 0$ and $A_0^{(0)} = A_0^{(1)} = 1$.
Assume that $1 \leq w < d$.
Since
\[
w - 2i \leq w < d_0,
\]
only the case $w - 2i = 0$ contributes to the calculation.
That is, $i = \frac{w}{2}$.
If so, we have
\[
\frac{w}{2} \leq \min\{w_1, w - w_1\},
\]
and $1 \leqslant w_1 \leqslant \frac{w}{2} < d_1$, which implies $A_{w_1}^{(1)} = 0$.
This in turn leads to $A_w = 0$.

\begin{proof}
Let $w \leq 2n$ be a nonnegative integer. We define the Hamming sphere
\[
S_w^{(2n)} = \{ \bm{c} \in \mathbb{F}_2^{2n} : w_H(\bm{c}) = w \},
\]
where $w_H(\bm{c})$ denotes the Hamming weight of $\bm{c}$.
To prove the theorem, we introduce a random vector $\bm{C}$, which is uniformly distributed over $S_w^{(2n)}$, i.e.,
\[
\Pr\{\bm{C} = \bm{c}\} =
\begin{cases}
\displaystyle \frac{1}{\binom{2n}{w}}, & \text{if } w_H(\bm{c}) = w, \\[6pt]
0, & \text{otherwise}.
\end{cases}
\]
We further write $\bm{C} = (\bm{C}^{(0)}, \bm{C}^{(1)})$, where $\bm{C}^{(i)} \in \mathbb{F}_2^n$, $i=0,1$.
Notice that in the construction of $\mathscr{C}$, the interleaver $\Pi$ is chosen uniformly at random.
Therefore, given $w_H(\bm{C}^{(1)}) = w_1$ and $w_H(\bm{C}^{(0)}+\bm{C}^{(1)}\Pi) = w_0$,
$\bm{C}^{(1)}$ is uniformly distributed over $S_{w_1}^{(n)}$ and $\bm{C}^{(0)}+\bm{C}^{(1)}\Pi$ is uniformly distributed over $S_{w_0}^{(n)}$.
This is guaranteed by the symmetry induced by the random permutation $\Pi$.

Obviously, the probability that $\bm{C}$ is a codeword of $\mathscr{C}$ is
\begin{align}\label{Eq:prob1_perm}
\Pr\{\bm{C}\in\mathscr{C}\} = \frac{A_w}{\binom{2n}{w}}.
\end{align}
On the other hand, $\bm{C}\in\mathscr{C}$ if and only if $\bm{C}^{(1)} \in \mathscr{C}_1$ and $\bm{C}^{(0)}+\bm{C}^{(1)}\Pi \in \mathscr{C}_0$.
Since $\Pi$ is uniformly distributed, this event occurs with probability
\begin{align}\label{Eq:prob2_perm}
\sum_{w_1=\max\{0,w-n\}}^{\min\{w,n\}}
   & \frac{\binom{n}{w_1}\binom{n}{w-w_1}}{\binom{2n}{w}}
   \cdot  \frac{A_{w_1}^{(1)}}{\binom{n}{w_1}} \notag \\
& \cdot \sum_{i=\max\{0,\,w-n\}}^{\min\{w_1,\,w-w_1\}}
   \frac{\binom{w-w_1}{i}\binom{n-w+w_1}{\,w_1-i\,}}{\binom{n}{\,w_1\,}}
   \, \frac{ A^{(0)}_{\,w-2i}}{\binom{n}{\,w-2i\,}}.
\end{align}

By equating \eqref{Eq:prob1_perm} and \eqref{Eq:prob2_perm}, we obtain the expression in \eqref{Eq:theorem}, which completes the proof.
\end{proof}

Notice that Theorem~\ref{Thm:main} provides a method to compute the ensemble weight distribution of $\mathscr{C}$ from those of $\mathscr{C}_0$ and $\mathscr{C}_1$, averaged over uniformly distributed interleavers $\Pi$. This recursion can be implemented alternatively in the manner shown below~\cite{Ma2015BMST}.
\begin{proof}[An equivalent recursion with proof]
Assume $\bm{u}$ and $\bm{v}$ are with weight $i$ and $j$ respectively.
Then the weight of codeword $(\bm{u} + \bm{v} \Pi,\bm{v})$ depends on the overlap between $\bm{u}$ and $\bm{v} \Pi$, which is denoted by $t$.
This occurs with probability
\[
  \frac{\binom{i}{t} \binom{n - i}{j - t}}{\binom{n}{j}}
\] 
and produces a codeword $(\bm{u} + \bm{v} \Pi,\bm{v})$ of weight $w = i + 2j - 2t$.
Therefore, the weight distribution of $\mathscr{C}$ can be expressed as
\[
  A_w = \sum_{i=0}^{n} \sum_{j=0}^{n} \sum_{t = 0}^{n} \mathbb{I}[i + 2j - 2t = w] \frac{\binom{i}{t} \binom{n - i}{j - t}}{\binom{n}{j}} A^{(0)}_{i} A^{(1)}_{j},
\]
where $\mathbb{I}[\cdot]$ is the indicator function.
Since $t$ is uniquely determined by $i,j,w$, we can eliminate the summation over $t$ and get
\[
  A_w = \sum_{\substack{i=0\\ (w - i) \text{ even}}}^{n} \sum_{j=0}^{n} \frac{\binom{i}{j - (w - i)/2} \binom{n - i}{(w - i)/2}}{\binom{n}{j}} A^{(0)}_{i} A^{(1)}_{j}.
\]
It can be verified that this recursion is equivalent to that given in Theorem 1.
\end{proof}

\end{theorem}

Notice that it is generally intractable to calculate the weight distribution of a specific $\mathscr{C}$ from those of $\mathscr{C}_0$ and $\mathscr{C}_1$. However, by introducing the random permutations, the ensemble weight distribution can be calculated by algorithms, say, presented in~\cite{ma2017systematic,chiu2020interleaved,yao2024}.
Compared with the existing algorithms over the polynomial rings, Theorem~1 provides a closed-form formula for the weight distribution. It becomes more convenient if only partial weight distribution is required, say, for truncated union bounds~\cite{ma2013new}.
In the following, we present two toy examples, where the component codes are invariant under all permutations. Hence the resulting weight distributions are also those of the specific codes.

\begin{example}
Let $\mathscr{C}_0 [3,1,3]$  and $\mathscr{C}_1 [3,2,2]$ be the repetition code and the single parity-check code of length $3$, respectively.
We have $A^{(0)}(x) = 1 + x^3$  and  $A^{(1)}(x) = 1 + 3x^2$.
Then, $A(x)$ of $\mathscr{C}$ is calculated as follows.

$A_0 = 1$, $A_1 = A_2 = 0$,
\begin{align*}
A_3 &= \sum_{w_1=\max\{0,3-3\}}^{\min\{3,3\}} \sum_{i=\max\{0,3-3\}}^{\min\{w_1,\,3-w_1\}}
      a(w_1,i)\,A^{(1)}_{w_1}\,A^{(0)}_{3-2i} \\
    &= a(0,0)A^{(1)}_0A^{(0)}_3 + a(2,0)A^{(1)}_2A^{(0)}_3 + a(2,1)A^{(1)}_2A^{(0)}_1 \\
    &= 1\cdot 1\cdot 1 + 1\cdot 3\cdot 1 + \tfrac{6}{9}\cdot 3\cdot 0 = 4,\\[6pt]
A_4 &= \sum_{w_1=\max\{0,4-3\}}^{\min\{4,3\}} \sum_{i=\max\{0,4-3\}}^{\min\{w_1,\,4-w_1\}}
      a(w_1,i)\,A^{(1)}_{w_1}\,A^{(0)}_{4-2i} \\
    &= a(2,1)A^{(1)}_2A^{(0)}_2 + a(2,2)A^{(1)}_2A^{(0)}_0 \\
    &= \tfrac{6}{9}\cdot 3\cdot 0 + 1\cdot 3\cdot 1 = 3,\\[6pt]
A_5 &= \sum_{w_1=\max\{0,5-3\}}^{\min\{5,3\}} \sum_{i=\max\{0,5-3\}}^{\min\{w_1,\,5-w_1\}}
      a(w_1,i)\,A^{(1)}_{w_1}\,A^{(0)}_{5-2i} \\
    &= a(2,2)A^{(1)}_2A^{(0)}_1 = 1\cdot 3\cdot 0 = 0,\\[6pt]
\mathrm{and}~ & \\
A_6 &= \sum_{w_1=\max\{0,6-3\}}^{\min\{6,3\}} \sum_{i=\max\{0,6-3\}}^{\min\{w_1,\,6-w_1\}}
      a(w_1,i)\,A^{(1)}_{w_1}\,A^{(0)}_{6-2i} \\
    &= a(3,3)A^{(1)}_3A^{(0)}_0 = 1\cdot 0\cdot 1 = 0.
\end{align*}

In summary,
$
A(x) = 1 + 4x^{3} + 3x^{4},
$
which can be verified by enumerating all the eight codewords in $\mathscr{C}$.
\end{example}

\begin{example}
Consider the $\mathrm{RM}(1,3)$ code, i.e., the $\mathscr{C}[8,4,4]$ code, whose weight distribution is
$
A(x) = 1 + 14x^{4} + x^{8}.
$
We now demonstrate how to compute $A(x)$ recursively. The encoding of $\mathrm{RM}(1,3)$ can be represented by a binary tree, as shown in Fig.~\ref{Fig:Polartree}.
Each node in the tree is associated with a vector. The active bits can take values $0$ or $1$, while the frozen bits are fixed to $0$.
The vector associated with a node can be calculated by the Plotkin construction as
\[
\bm{v} = (\bm{v}_0 + \bm{v}_1 \Pi, \bm{v}_1),
\]
where $\bm{v}_0$ and $\bm{v}_1$ are vectors associated with the left child and the right child of the node, respectively.
The codeword is thus the vector associated with the root node. This corresponds exactly to a multilevel Plotkin construction, which can be employed to recursively calculate the weight distribution at the root node. This procedure is illustrated in Fig.~\ref{Fig:Polar}, where the weight distributions at the leaves are initialized as $1$ for a frozen leaf and as $1+X$ for an active leaf.

\begin{figure}[t]
  \centering
  \includegraphics[width=0.45\textwidth]{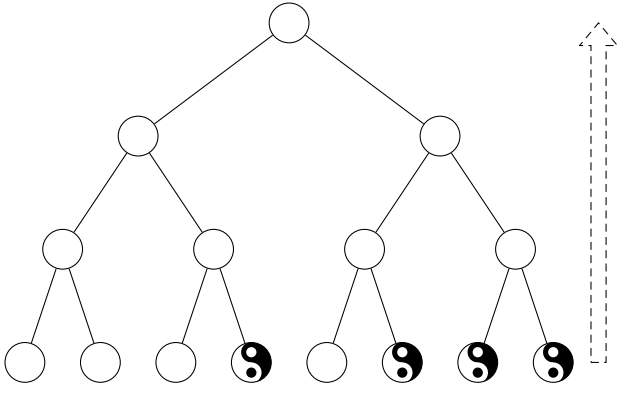}\\
  \caption{The tree structure of the RM code $\mathscr{C}[8,4,4]$. Hollow circles at the leaves denote frozen positions (fixed to $0$), while the Taiji (yin–yang) symbols denote active (information) positions (which can be $0$ or $1$). Encoding proceeds from the leaves to the root.}\label{Fig:Polartree}
\end{figure}

\begin{figure}[t]
  \centering
  \includegraphics[width=0.45\textwidth]{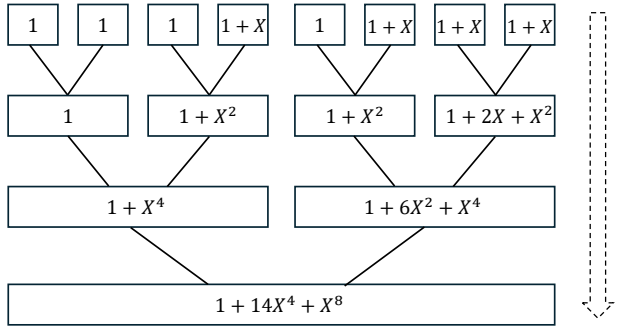}\\
  \caption{An illustration of the recursive calculation of the weight distribution of the RM code $\mathscr{C}[8,4,4]$.
  The weight distributions at the leaves are initialized as $1$ for a frozen leaf and to $1+X$ for an active leaf, and the calculation proceeds from top to bottom.}\label{Fig:Polar}
\end{figure}

\end{example}

\section{Conclusion}
In this note, we have essentially developed an algorithm~(with a computational complexity of polynomial order) to calculate the weight distribution of a concatenated code ensemble based on the Plotkin construction from those of its component codes.

\section*{Acknowledgment}
The author would like to thank Dr.~Qianfan Wang and Dr.~Jifan Liang for their helpful discussions. 
The author also thanks Dr.~Qianfan Wang for his assistance in preparing this note.
The alternative recursion for Theorem~\ref{Thm:main} is provided and typed by Dr.~Jifan Liang and Dr.~Xinyuanmeng Yao.
\bibliographystyle{IEEEtran}
\bibliography{bibliofile}
\end{document}